\documentclass[runningheads,envcountsame]{llncs}
\usepackage{preamble}
\usepackage{macros}
\usepackage[T1]{fontenc}
\usepackage{graphicx}
\usepackage{color}

\newcommand{\cld}{\mathsf{child}}
\newcommand{\pnd}{\mathsf{pendant}}
\newcommand{\Oh}{\mathcal{O}}

\begin{document}
\title{Kernelization for Finding Lineal Topologies (Depth-First Spanning Trees) with Many or Few Leaves\thanks{The research leading to these results has received funding from the Research Council of Norway via the
projects (PCPC) (grant no. 274526) and BWCA (grant no. 314528).
}}
\titlerunning{Kernelization for Finding Lineal Topologies with Many or Few Leaves}
\author{Emmanuel Sam\inst{1}\orcidID{0000-0001-7756-0901} \and
Benjamin Bergougnoux\inst{2}\orcidID{0000-0002-6270-3663} \and
Petr A. Golovach \inst{1}\orcidID{0000-0002-2619-2990} \and 
Nello Blaser \inst{1}\orcidID{0000-0001-9489-1657}}
\authorrunning{E. Sam, B. Bergougnoux, P. Golovach, N. Blaser}
\institute{Department of Informatics, University of Bergen, Norway 
\email{\{emmanuel.sam,petr.golovach, nello.blaser\}@uib.no}\\
%\url{http://www.springer.com/gp/computer-science/lncs} 
\and
Institute of Informatics, University of Warsaw, Poland\\
\email{benjamin.bergougnoux@mimuw.edu.pl}}
\maketitle              % typeset the header of the contribution
\begin{abstract}
For a given graph $G$, a depth-first search (DFS) tree $T$ of $G$ is an $r$-rooted spanning tree such that every edge of $G$ is either an edge of $T$ or is between a \textit{descendant} and an \textit{ancestor} in $T$.
A graph $G$ together with a DFS tree is called a \textit{lineal topology} $\mathcal{T} = (G, r, T)$.  Sam et al. (2023) initiated study of the parameterized complexity of the \textsc{Min-LLT} and \textsc{Max-LLT} problems which ask, given a graph $G$ and an integer $k\geq 0$, whether $G$ has a DFS tree with at most $k$ and at least $k$ leaves, respectively. Particularly, they showed that for the dual parameterization, where the tasks are to find DFS trees  with at least $n-k$ and at most $n-k$ leaves, respectively, these problems are fixed-parameter tractable when parameterized by $k$. However, the proofs were based on Courcelle's theorem, thereby making the running times a tower of exponentials. We prove that both problems admit polynomial kernels with $\Oh(k^3)$ vertices. In particular, this implies FPT algorithms running in  $k^{\Oh(k)}\cdot n^{O(1)}$ time.
We achieve these results by making use of a $\Oh(k)$-sized vertex cover structure associated with each problem. This also allows us to demonstrate polynomial kernels for \textsc{Min-LLT} and \textsc{Max-LLT}  for the structural parameterization by the vertex cover number.

\end{abstract}

\newpage
\section{Introduction}
\label{intro}
Depth-first search (DFS) is a well-known fundamental technique for visiting the vertices and exploring the edges of a graph \cite{cormen2009,tarjan1971}. For a given connected undirected graph with vertex set $V(G)$ and edge set $E(G)$, DFS explores $E(G)$ by always choosing an edge incident to the most recently discovered vertex that still has unexplored edges. A selected edge, either leads to a new vertex or a vertex already discovered by the search. The set of edges that lead to a new vertex during the DFS define an $r$-rooted spanning tree $T$ of $G$, called a \textit{depth-first spanning} (DFS) tree, where $r$ is the vertex from which the search started. This tree $T$ has the property that each edge that is not in $T$ connects an ancestor and a descendant of $T$. %The edges in $E(G) \setminus E(T)$ are called \textit{back edges}. 
All rooted spanning trees of a finite graph with this property, irrespective of how they are computed, such as a \textit{Hamiltonian path}, are generalized as \textit{tr\'{e}maux trees} \cite{DEFRAYSSEIX2012279}. 
Given a graph $G$ and a DFS tree $T$ rooted at a vertex $r \in V(G)$, it is easy to see that the family $\mathcal{T}$ of subsets of $E(G)$ induced by the vertices in all subtrees of $T$ with the same root $r$ as $T$ constitute a topology on $E(G)$. For this reason, the triple $(G, T, r)$ has been referred to as the \textit{lineal topology} (LT) of $G$ in \cite{Sam2023a}. 
Many existing applications of DFS and DFS trees --- such as planarity testing and embedding \cite{de2008tremaux,hopcroft1974efficient}, finding connected and biconnected components of undirected graphs \cite{HopcroftTarjan1973}, bipartite matching \cite{HopcroftK73}, and graph layout \cite{BIEDL2001167} --- only require one to find an arbitrary DFS tree of the given graph, which can be done in time $O(n + m)$, where $n$ and $m$ are the number of vertices and edges of the graph.

An application of a DFS tree, noted by  Fellows et al. \cite{fellows1988}, that calls for a DFS tree  with minimum height is the use of DFS trees to structure the search space of backtracking algorithms for solving \textit{constraint satisfaction problems} \cite{freuder1985}. 
This motivated the authors to study the complexity of finding DFS trees of a graph $G$ that optimize or near-optimize the maximum length or minimum length of the root-to-leaf paths in the DFS trees of $G$.
They showed that the related decision problems are NP-complete and do not admit a polynomial-time absolute approximation algorithm unless P = NP. 

In this paper, we look at the \textsc{Minimum Leafy LT} (\textsc{Min-LLT}) and \textsc{Maximum Leafy LT} (\textsc{Max-LLT}) problems introduced by Sam et al. \cite{Sam2023a}. Given a graph $G$ and an integer $k\geq 0$, \textsc{Min-LLT} and \textsc{Max-LLT} ask whether $G$ has a DFS tree with at most $k$ and at least $k$ leaves, respectively. These two problems are related to the well-known NP-complete \textsc{Minimum Leaf Spanning Tree (Min-LST)} and \textsc{Maximum Leaf Spanning Tree (Max-LST)} \cite{gareyJohnson1990,RAHMAN200593}.

Sam et al.  \cite{Sam2023a} proved that \textsc{Min-LLT} and \textsc{Max-LLT} are NP-hard.
Moreover, they proved that when parameterized by $k$, \textsc{Min-LLT} is para-NP-hard and \textsc{Max-LLT} is W[1]-hard.
They also considered the \say{dual} parameterizations, namely, \textsc{Dual Min}-LLT and \textsc{Dual Max}-LLT, where the tasks are to find DFS trees with at least $n-k$ and at most $n-k$ leaves, respectively.
They proved that \textsc{Dual Min}-LLT and \textsc{Dual Max}-LLT are both FPT parameterized by $k$.
These FPT algorithms are, however, based on Courcelle's theorem \cite{COURCELLE199012}, which relates the expressibility of a graph property in \textit{monadic second order} (MSO) logic to the existence of an algorithm that solves the problem in FPT-time with respect to \textit{treewidth} \cite{Nesetril2012}. As a by-product, their running times have a high exponential dependence on the treewidth and the length of the MSO formula expressing the property.

\subsection{Our Results}
\label{subsec:result}
We prove that \textsc{Min-LLT} and  \textsc{Max-LLT} admit polynomial kernels when parameterized by the \textit{vertex cover number} of the given graph. Formally, we prove the following theorem.
\begin{theorem}\label{thm:main:vc}
    \textsc{Min-LLT} and \textsc{Max-LLT} admit kernels with $\Oh(\tau^3)$ vertices when parameterized by the vertex cover number $\tau$ of the input graph. 
\end{theorem}
Based on these kernels, we show that \textsc{Dual Min-LLT}, and \textsc{Dual Max-LLT} admit polynomial kernels parameterized by $k$.
\begin{theorem}\label{thm:main:dual}
 \textsc{Dual Min-LLT} and \textsc{Dual Max-LLT} admit kernels with $\Oh(k^3)$ vertices.
 \end{theorem}
This last result follows from a win-win situation as either (1)~the input graph has a large \textit{vertex cover} in terms of $k$ and, consequently, both problems are trivially solvable or (2)~the input graph has a small vertex cover, and we can use Theorem~\ref{thm:main:vc}.
Finally, we use our polynomial kernels to prove that \textsc{Dual Min-LLT}, and \textsc{Dual Max-LLT} admit %single-exponential time 
FPT algorithms parameterized by $k$ with low exponential dependency. 
\begin{theorem}\label{thm:main:FPT}
    \textsc{Dual Min-LLT} and \textsc{Dual Max-LLT} can be solved in $k^{\Oh(k)}\cdot n^{\Oh(1)}$ time.  
 \end{theorem}
As the previously known FPT algorithm for each of these problems was based on Courcelle's theorem, our algorithms are the first FPT-algorithms constructed explicitly.

\subsection{Related Results}
Lu and Ravi \cite{Lu96thepower} proved that the \textsc{Min-LST}, problem has no constant factor approximation unless $P = NP$. From a parameterization point of view, Prieto et al.\cite{elenaP2003} showed that this problem is $W[P]$-hard parameterized by the solution size $k$. The \textsc{Max-LST} problem is, however, FPT parameterized by $k$ and has been studied extensively \cite{bonsma2003faster,bonsma2008,vladimirEst2005,fellows1992well,michael2000coordinatized}.

\textsc{Dual Min-LLT} is related to the well-studied $k$-\textsc{Internal Spanning Tree} problem \cite{fomin2009MaxInt,elenaP2003}, which asks to decide whether a given graph admits a spanning tree with at most $n-k$ leaves (or at least $k$ internal vertices). Prieto et al.\cite{elenaP2003} were the first to show that the natural parameterized version of $k$-\textsc{Internal Spanning Tree} has a $\Oh^*(2^{k\log{}k})$-time FPT algorithm and a $\Oh(k^3)$-vertex kernel. Later, the kernel was improved to $\Oh(k^2)$, $\Oh(3k)$, and $\Oh(2k)$ by Prieto et al., Fomin et al.\cite{fomin2009MaxInt}, and Li et al. \cite{wenjunLi2017} respectively. The latter authors also gave what is now the fastest FPT algorithm for $k$-\textsc{Internal Spanning Tree}, which runs in $\Oh^*(4^k)$ time. 

An \textit{independency tree} (IT) is a variant of a spanning tree whose leaves correspond to an independent set in the given graph. Given a connected graph on $n \ge 3$, $G$ has no IT if it has no DFS tree in  which the leaves and the root are pairwise nonadjacent in $G$ \cite{BOHME1997219}. 
From a parameterization point of view, the \textsc{Min Leaf IT (Internal)} and \textsc{Max Leaf IT (Internal)} problems \cite{CASEL20202}, which ask, given a graph $G$ and an integer $k \ge 0$, whether  $G$ has an IT with at least $k$ and at most $k$ internal vertices, respectively, are related to \textsc{Dual Min-LLT} and \textsc{Dual Max-LLT}, respectively.  Casel et al. \cite{CASEL20202} showed that, when parameterized by $k$, \textsc{Min Leaf IT (Internal)} has an $\Oh^*(4^{k})$-time algorithm and a $2k$ vertex kernel. They also proved that \textsc{Max Leaf IT (Internal)} parameterized by $k$ has a $\Oh^*(18^{k})$-time algorithm and a $\Oh(k2^k)$-vertex kernel, but no polynomial kernel unless the polynomial hierarchy collapses to the third level. Their techniques, however, do not consider the properties of a DFS tree and, therefore, do not work for our problems.

\subsection{Organization of the paper}

Section \ref{prelims} contains basic terminologies relevant to graphs, DFS trees, and parameterized complexity necessary to understand the paper.  In section \ref{sectn:kernel}, we first prove a lemma about how, given a graph $G$ and a vertex cover of $G$, the internal vertices of any spanning tree of $G$ relate to  the given vertex cover. We then use this lemma to demonstrate a polynomial kernel for \textsc{Min-LLT} and \textsc{Max}-LLT for the structural parameterization by the vertex cover number of the graph. This is followed by the kernelization algorithms for \textsc{Dual Min-LLT} and \textsc{Dual Max-LLT} parameterized by $k$. In section \ref{sectn:fpt_algorithms}, we devise FPT algorithms for \textsc{Dual Min-LLT} and \textsc{Dual Max-LLT} based on their polynomial kernels. Finally, we conclude the paper in section \ref{concln} with remarks concerning future studies.
%%----------------------------------------------------------------------------
%%SECTION 2
\section{Preliminaries}
\label{prelims}
%\textcolor{red}{Here, we present graph-theoretic preliminaries, definitions and notations related to lineal topology, parameterized complexity, treewidth, and Monadic Second Order Logic, needed to understand the work herein. Finally, we present the problem definitions.}
%In this section, we outline some of the standard concepts of graph theory and the notations needed to understand the work herein.
%We also define formally the topology in the terminology of how a lineal spanning tree defines a topology on the edge set of a graph that the back-edge structure induced by a lineal spanning tree general topology,
%A lineal topology has to do with the topology of the back-edge structure induced by a lineal spanning tree $T$.
We consider only simple finite graphs. We use $V(G)$ and $E(G)$ to denote the sets of vertices and edges, respectively, of a graph $G$.
For a graph $G$, we denote the number of vertices $|V(G)|$ and the number of edges $|E(G)|$ of $G$ by $n$ and $m$, respectively, if this does not create confusion. 
%We use $uv$ %instead of $\{u,v\}$ to denote an edge in $E(G)$.
  For any vertex $v \in V(G)$, the set $N_G(v)$ denotes the neighbors of $v$ in $G$ and $N_{G}[v]$ denotes its \textit{closed neighborhood} $N_G(v) \cup \{v\}$ in $G$. 
 For a set of vertices $X\subseteq V$, $N_G(X)=\big(\bigcup_{v\in X}N_G(v)\big)\setminus X$.
% The \textit{degree} of $v$, denoted by $d_G(v)$, is the number of elements in $N_G(v)$, and the degree of $G$, denoted by $\Delta(G)$, is defined as $\max_{v \in V(G)}d_G(v)$. For each of the above-given notations, 
We omit the $G$ in the subscript if the graph is clear from the context. 
For a vertex $v$, its \emph{degree} is $d_G(v)=|N_G(v)|$.
Given any two graphs $G_1=(V_1, E_1)$ and $G_2=(V_2, E_2)$, if $V_1 \subseteq V_2$ and $E_1 \subseteq E_2$ then $G_1$ is a \textit{subgraph} of $G_2$, denoted by $G_1 \subseteq G_2$.
If $G_1$ contains all the edges $uv \in E_2$ with $u,v \in V_1$, then we say $G_1$ is an \textit{induced subgraph} of $G_2$, or $V_1$ induces $G_1$ in $G_2$, denoted by $G[V_1]$.
If $G_1$ is such that it contains every vertex of $G_2$, i.e., if $V_1 = V_2$ then $G_1$ is a \textit{spanning subgraph} of $G_2$.
Given a set of vertices $X \subseteq V(G)$, we express the induced subgraph $G[V(G)\setminus X]$ as $G-X$. If $X=\{x\}$, we write $V(G)\setminus x$ instead of $V(G)\setminus \{x\}$ and $G-x$ instead of $G-\{x\}$. Given a graph $G$, a set of vertices $S \subseteq V(G)$ is a \textit{vertex cover} of $G$ if, for every edge $uv \in E(G)$, either $u \in S$ or $v \in S$; the \emph{vertex cover number} of $G$, denoted by $\tau(G)$, is the minimum size of a vertex cover. A set $Y \subseteq V(G)$ is called an \textit{independent set}, if for every vertex pair $u, v \in Y$, $uv \notin E(G)$. A \textit{matching} $M$ in a given graph $G$ is a set of edges, no two of which share common vertices. 
%We say that $M$ is a \textit{perfect matching} if every vertex in $V(G)$ is incident with one of the edges in $M$, and $M$ is a \textit{maximal matching} if it is not a subset of any other matching of $G$. \
A \textit{pendant vertex} is a vertex with degree one.

For definitions of basic tree terminologies including root, child, parent, ancestor, and descendant, we refer the reader to  \cite{DiestelGraphTheory}. Given a graph $G$, we denote a spanning tree of $G$ rooted at a vertex $r \in V(G)$ by $(T,r)$. When there is no ambiguity, we simply use $T$ instead of $(T,r)$. For a rooted tree $T$, a vertex $v$ is a \emph{leaf} if it has no descendants and $v$ is an \emph{internal} vertex if otherwise. 
A spanning tree $T$ with a root $r$ is a \emph{DFS tree rooted in $r$} if for very every edge $uv\in E(G)$, either $uv\in E(T)$, or $v$ is a descendant  of $u$ in $T$, or $u$ is a descendant  of $v$ in $T$. Equivalently, 
$T$ is a DFS tree if it can be produced by the classical depth-first search (DFS) algorithm~\cite{cormen2009}.  We say that a path $P$ in a rooted tree $T$ is a \emph{root-to-leaf} path if one of its end-vertices is the root and the other is a leaf of $T$.

\medskip
Now we review some important concepts of Parameterized complexity (PC) relevant to the work reported herein. For more details about PC, we refer the reader to \cite{cyganPC,downeyFellow2013}. 

\begin{definition}[Parameterized problem]
Let $\Sigma$ be a fixed finite alphabet. A \textit{parameterized problem} is a language $P \subseteq \Sigma^{\ast} \times \N$. Given an instance $(x,k) \in \Sigma^{\ast} \times \N$ of a parameterized problem, $k \in \N$ is called the \textit{parameter}, and the task is to determine whether $(x,k)$ belongs to $P$. A parameterized problem $P$ is classified as \textit{fixed-parameter tractable} (FPT) if there exists an algorithm that answers the question \say{$(x,k) \in P$?} in time $f(k)\cdot poly(|x|)$, where $f: \N \rightarrow \N$ is a computable function. 
\end{definition}
\begin{definition}
A kernelization algorithm, or simply a kernel, for a parameterized problem $P$ is a function $\phi$ that maps an instance $(x,k)$ of $P$ to an instance $(x',k')$ of $P$ such that the following properties are satisfied:
\begin{enumerate}
    \item \label{itm:1}$(x,k) \in P$ if and only if $(x',k') \in P$,
    \item \label{itm:2} $k'+|x'| \le g(k)$ for some computable function $g:\N \rightarrow \N$, and
    \item \label{itm:3} $\phi$ is computable in time polynomial in $|x|$ and $k$.
\end{enumerate}
\end{definition}
If the upper-bound $g(\cdot)$ of the kernel (Property~\ref{itm:2}) is polynomial (linear) in terms of the parameter $k$, then we say that $P$ admits a polynomial (linear) kernel. 
It is common to write a kernelization algorithm as a series of reduction rules. A \emph{reduction rule} is a polynomial-time algorithm that transform an instance $(x,k)$ to an equivalent instance $(x',k')$ such that Property~ \ref{itm:1} is fulfilled. Property~\ref{itm:1} is referred to as the \textit{safeness} or \textit{correctness} of the rule.

\section{Kernelization}
\label{sectn:kernel}
In this section, we demonstrate polynomial kernels for \textsc{Dual Min-LLT} and \textsc{Dual Max-LLT}. But first, we show that \textsc{Min-LLT} and \textsc{Max-LLT} admit polynomial kernels when parameterized by the vertex cover number of the input graph. The following simple lemma is crucial for our kernelization algorithms.

\begin{lemma}\label{upperB_internal}
Let $G$ be a connected graph and let $S$ be a vertex cover of $G$. Then every rooted spanning tree $T$ of $G$ has at most $2|S|$ internal vertices and at most $|S|$ internal vertices are not in $S$.
\end{lemma}

\begin{proof}
Let $T$ be a rooted spanning tree tree of $G$ with a set of internal vertices $X$.
For every vertex $v$ of $T$, we denote by $\cld(v)$ the set of its childred in $T$.
For each internal vertex $v$ of $T$, we have $\cld(v)\neq \emptyset$ and if $v\notin S$, then $\cld(v)\subseteq S$ because $S$ is a vertex cover of $G$.
Moreover, for any distinct internal vertices $u$ and $v$ of $T$,  $\cld(u)\cap \cld(v)=\emptyset$.
Given $X\setminus S= \{v_1,\dots,v_t\}$, we deduce that $\cld(v_1),\dots,\cld(v_t)$ are pairwise disjoint and non-empty subsets of $S$. 
We conclude that $|X\setminus S|\leq |S|$ and $|X|\leq  2 |S|$.
\qed
\end{proof}
We also use the following folklore observation.

\begin{observation}\label{obs:int}
The set of internal vertices of any DFS tree $T$ of a connected graph $G$ is a vertex cover of $G$.
\end{observation}

\begin{proof}
To see the claim, it is sufficient to observe that any leaf of a DFS tree $T$ is adjacent in $G$ only to its  ancestors, that is, to internal vertices.
\qed
\end{proof}
We use Lemma~\ref{upperB_internal} to show that, given a vertex cover, we can reduce the size of the input graph for both  \textsc{Min-LLT} and \textsc{Max-LLT}.

\begin{lemma}\label{lem:vc}
There is a polynomial-time algorithm that, given a connected graph $G$ together with a vertex cover $S$ of size $s$, outputs a graph $G'$ with at most  $s^2(s-1)+3s$ vertices such that for every integer $t\geq 0$,
$G$ has a DFS tree with exactly $t$ internal vertices if and only if $G'$ has a DFS tree with exactly $t$ internal vertices.
\end{lemma}
 
\begin{proof}
Let $G$ be a connected graph and let $S$ be a vertex cover of $G$ of size $s$. As the lemma is trivial if $s=0$, we assume that $s\geq 1$.
Denote $I=V(G)\setminus S$; note that $I$ is an independent set. 
We apply the following two reduction rules to reduce the size of $G$. 

\medskip
The first rule reduces the number of pendant vertices. To describe the rule, denote by $\pnd(v)$ for $v\in S$ the set of pendant vertices of $I$ adjacent to $v$.

\SetAlgoProcName{Rule}{Rule}
\begin{procedure}[bth]
	\SetAlgoLined
	\ForEach{$v\in S$}
	{
		\lIf{$|\pnd(v)| > 2$}{delete all but two vertices in $\pnd(v)$ from $G$}
	}
	\caption{1()}
	\label{dual_min_R1}
\end{procedure}

To see that Rule~\ref{dual_min_R1} is safe, denote by $G'$ the graph obtained from $G$ by the application of the rule. Notice that for every $v\in S$, at most one vertex of $\pnd(v)$ is the root and the other vertices are leaves that are children of $v$ in any rooted spanning tree $T$ of $G$. 

Let $T$ be a DFS tree of $G$ rooted in $r$ with $t$ internal vertices. Because for every $v\in S$, the vertices of $\pnd(v)$ have the same neighborhood in $G$ and Rule~\ref{dual_min_R1}  does not delete all the vertices of $\pnd(v)$, we can assume without loss of generality that $r\in V(G')$. 
Let $T'=T[V(G')]$. Because the deleted vertices are leaves of $T$, we have that $T'$ is a tree and, moreover, $T'$ is a DFS tree of $G'$ rooted in $r$. Clearly, each internal vertex of $T'$ is an internal vertex of $T$.
Let $v\in S$ be a vertex such that $|\pnd(v)|>2$. Then $v$ has a pendant neighbor $u\neq r$ in $G'$ and $u$ should be a child of $v$ in $T'$. Thus, $v$ is an internal vertex of $T'$. This implies that every leaf $v$ of $T'$ is not adjacent to any vertex of $V(G)\setminus V(G')$ in $G$. Hence, $v$ is a leaf of $T$. Because the deleted vertices are leaves of $T$, we obtain that a vertex $v\in V(G)$ is an internal vertex of $T$ if and only if $v$ is an internal vertex of $T'$. Then $T$ and $T'$ have the same number of internal vertices.

For the opposite direction, let $T'$ be a DFS tree of $G'$ rooted in $r$ with $t$ internal vertices. We construct the tree $T$ from $T'$ by adding each deleted vertex $u$ as a leaf to $T'$: if $u\in V(G)\setminus V(G')$, then $u\in\pnd(v)$ for some $v\in S$ and we add $u$ as a leaf child of $v$. Because the deleted vertices are pendants, we have that $T$ is a DFS tree of $G$. Observe that each internal vertex of $T'$ remains internal in $T$. 
In the same way as above, we observe that a vertex  $v\in S$ with $|\pnd(v)|>2$ cannot be a leaf of $T'$, because $v$ has a pendant neighbor in $G'$ distinct from $r$ that should be a child of $v$. Hence, 
every leaf $v$ of $T'$  is not adjacent to any vertex of $V(G)\setminus V(G')$ in $G$ and, therefore, is a leaf of $T$. Since  the deleted vertices are leaves of $T$, we obtain that a vertex $v\in V(G)$ is an internal vertex of $T$ if and only if $v$ is an internal vertex of $T'$. Thus, $T$ and $T'$ have the same number of internal vertices. This concludes the safeness proof.

\medskip
The next rule is used to reduce the number of nonpendant vertices of $I$.  For each pair of vertices $u, v \in S$, we use \emph{common neighbor} of $u$ and $v$ to refer to a vertex $w \in I$ that is adjacent to both $u$ and $v$ and denote by $W_{uv}$ the set of common neighbors of  $u$ and $v$.  Rule \ref{dual_min_R2}  is based on the observation that if the size of $W_{uv}$ for any vertex pair $u,v \in S$ is at least $2s+1$, then it follows from Lemma~\ref{upperB_internal} that every spanning tree $T$ contains at most $s$ internal vertices and at least $s+1$ leaves from $W_{uv}$. We prove that it is enough to keep at most $2s$ vertices from $W_{uv}$ for each $u,v\in S$.
%\noindent

\begin{procedure}
\ForAll{pairs $\{u,v\}$ of distinct vertices of $S$}{Label $\max\{|W_{uv}|,2s\}$ vertices in $W_{uv}$\;} 
Delete the unlabeled vertices of $I$ with at least two neighbors in $S$ from $G$.
\caption{2()}
\label{dual_min_R2}
\end{procedure}

\newcommand{\qedclaim}{\hfill $\blacklozenge$}

To show that Rule~\ref{dual_min_R2} is safe, let $x\in I$ be a vertex with at least two neighbors in $S$ which is not labeled by Rule~\ref{dual_min_R2}. Let $G'=G-x$. We claim that $G$ has a DFS tree with exactly $t$ internal vertices if and only if $G'$ has a DFS tree with exactly $t$ internal vertices. 

We use the following auxiliary claim, the proof of which can be found in Appendix \ref{appendix:A}.

\begin{claimperso} \label{cl:aux}
\text{}
\begin{itemize}
\item[(i)] For any DFS tree $T$ of $G$,  the vertices of $N_G(x)$ are vertices of a root-to-leaf path of $T$.
\item[(ii)] For any DFS tree $T'$ of $G'$, the vertices of $N_G(x)$ are vertices of a root-to-leaf path of $T'$.
\item[(iii)] For any DFS tree $T'$ of $G'$, every vertex of $N_G(x)$ is an internal vertex of $T'$.
\end{itemize} 
\end{claimperso}

% \begin{proof}[Proof of Claim~\ref{cl:aux}]
% We show (i) by contradiction. Assume that there are $u,v\in N_G(x)$ such that the lowest common ancestor $w$ of these vertices is distinct from $u$ and $v$. Because $x$ is not labeled by Rule~\ref{dual_min_R2}, $|W_{uv}|>2s$. Hence, by Lemma~\ref{upperB_internal} , there is a vertex $z\in W_{uv}$ such that $z$ is a leaf of $T$. However, any leaf in a DFS tree of $T$ can be adjacent only to its ancestors in $T$. This contradiction proves the claim.

% We use exactly the same arguments to prove (ii) by replacing $T$ by $T'$ and observing that $S$ is a vertex cover of $G'$.

% To show (iii), let $T'$ be a DFS tree with a root $r$. By (ii), there is a leaf $y$ such that the vertices of $N_G(x)$ are vertices of the $(r,y)$-path in $T'$. Observe that $y$ may be not unique. We prove that $y\notin N_G(x)$. For the sake of contradiction, assume that $x$ and $y$ are adjacent. Because $d_G(x)\geq 2$, $x$ has a neighbor $u\neq x$. Because $x$ is not labeled by Rule~\ref{dual_min_R2}, $|W_{uy}|>2s$. By Lemma~\ref{upperB_internal}, we obtain that there is $v\in W_{uy}$ that is a leaf of $T'$. We have that $vy\in E(G')$ but two leaves of a DFS tree cannot be adjacent; a contradiction. This proves that $y\notin N_G(x)$ and concludes the proof of the claim.  
% \qedclaim
% \end{proof}

We use Claim~\ref{cl:aux} to show the following property.

\begin{claimperso}\label{cl:leaf-x}
If $G$ has a DFS tree with $t$ internal vertices, then $G$ has a DFS tree $T$ with $t$ internal vertices such that $x$ is a leaf of $T$.
\end{claimperso}

\begin{proof}[Proof of Claim~\ref{cl:leaf-x}]
Let $T$ be a DFS tree of $G$ with a root $r$ that has exactly $t$ internal vertices. We prove that if $x$ is an internal vertex of $T$, then $T$ can be modified in such a way that $x$ would become a leaf. Observe that by Claim~\ref{cl:aux} (i), $x$ has a unique child $v$ in $T$. We have two cases depending on whether  $x=r$ or has a parent $u$.    

Suppose first that $x=r$. By Claim~\ref{cl:aux}, the neighbors of $x$ in $G$ are vertices of some root-to-leaf path of $T$. Let $u$  be the neighbor of $x$ at maximum distance from $r$ in $T$.
Because $d_G(x)\geq 2$, $u\neq v$.   Since  $x$ is not labeled by Rule~\ref{dual_min_R2}, $|W_{uv}|>2s$. By Lemma~\ref{upperB_internal}, there are
at least $s+1$ vertices $W_{uv}$ that are leaves of $T$. These leaves have their parents in $S$ which has size $s$. By the pigeonhole principle, there are distinct leaves $w,w'\in W_{uv}$ with the same parent. We rearrange $T$ by making $w$ a root with the unique child $v$ and making $x$ a leaf with the parent $u$. Denote by $T'$ the obtained tree. 

Because $x$ is adjacent to $u$ and some of its ancestors in $T$ and $w$ is adjacent only to some of its ancestors in $T$, we conclude that $T'$ is a feasible DFS tree.  Notice that $w$ which was a leaf of $T$ became an internal vertex of $T'$ and $x$ that was an internal vertex is now a leaf. Because $x$ is a leaf of $T'$, we have that $T''=T'-x$ is a DFS tree of $G'$ rooted in $w$. By Claim~\ref{cl:aux} (iii), $u$ is an internal vertex of $T''$. This implies that $u$ is an internal vertex of both $T$ and $T'$. Since the parent of $w$ in $T$ has $w'\neq w$ as a child, we also have that $w$ is an internal vertex of both $T$ and $T'$. Therefore, $T$ and $T'$ have the same number of internal vertices. This proves that $G$ has a DFS tree $T'$ with $t$ internal vertices such that $x$ is a leaf of $T'$.

Assume now that $x$ has a parent $u$ in $T$. By Claim~\ref{cl:aux}, the neighbors of $x$ in $G$ are vertices of some root-to-leaf path of $T$. Denote by  $v'$  be the neighbor of $x$ at maximum distance from $r$ in $T$; it may happen that $v'=v$. As $x$ is not labeled by Rule~\ref{dual_min_R2}, $|W_{uv}|>2s$. Then by Lemma~\ref{upperB_internal}, there are
at least $s+1$ vertices $W_{uv}$ that are leaves of $T$. These leaves have their parents in $S$ which has size $s$. By the pigeonhole principle, there are distinct leaves $w,w'\in W_{uv}$ with the same parent. We rearrange $T$ by making $w$ a child of $u$ and a parent of $v$ and making $x$ a leaf with the parent $v'$. Denote by $T'$ the obtained tree. 

Because $x$ is adjacent to $v'$ and some of its ancestors in $T$ and $w$ is adjacent only to some of its ancestors in $T$, including $u$ and $v$, we have that $T'$ is a feasible DFS tree.  Notice that $w$ was a leaf of $T$ and is now an internal vertex of $T'$, while $x$ was an internal vertex in $T$ and is now a leaf in $T'$. Because $x$ is a leaf of $T'$, we have that $T''=T'-x$ is a DFS tree of $G'$ rooted in $w$. By Claim~\ref{cl:aux} (iii), $v'$ is an internal vertex of $T''$. Therefore, $v'$ is an internal vertex of both $T$ and $T'$. Since the parent of $w$ in $T$ has $w'\neq w$ as a child, we also have that $w$ is an internal vertex of both $T$ and $T'$. Thus, $T$ and $T'$ have the same number of internal vertices. We obtain that $G$ has a DFS tree $T'$ with $t$ vertices such that $x$ is a leaf of $T'$. This concludes the proof.
\qedclaim
\end{proof}

Now we are ready to proceed with the proof that $G$ has a DFS tree with exactly $t$ internal vertices if and only if $G'$ has a DFS tree with exactly $t$ internal vertices. 

For the forward direction, let $T$ be a DFS tree of $G$ with $t$ internal vertices. By Claim~\ref{cl:leaf-x}, we can assume that $x$ is a leaf of $T$.
Let $T'=T-x$. Because $x$ is a leaf of $T$, $T'$ is a DFS tree of $G'$. %We show that $T$ and $T'$ have the same sets of internal vertices.  
Let $u$ be the parent of $x$ in $T$. Because $u$ is adjacent to $x$ in $G$, we have that $u$ is an internal vertex of $T'$ by Claim~\ref{cl:aux} (iii). 
This means that the number of internal vertices of $T$ and $T'$ is the same, that is, $G'$ has a DFS tree with $t$ vertices.

For the opposite direction, let $T'$ be a DFS tree of $G'$ with $t$ internal vertices with a root $r$. By Claim~\ref{cl:aux} (ii),  the neighbors of $x$ in $G$ are vertices of some root-to-leaf path in $T'$. Let $v$ be the neighbor of $x$ at maximum distance from $r$ in $T'$. We construct $T$ by making $x$ a leaf with the parent $v$.  Because $x$ is adjacent in $G$ only to $v$ and some of  its ancestors in $T'$, $T$ is a DFS tree. By Claim~\ref{cl:aux}(iii), $v$ is an internal vertex of $T'$. Therefore, $T'$ and $T$ have the same set of internal vertices. We obtain that  $G$ has a DFS tree with $t$ vertices. This concludes the proof of our claim.

Recall that $G'$ was obtained from $G$ by deleting a single unlabeled vertex $x\in I$ of degree at least two. Applying the claim that $G$ has a DFS tree with exactly $t$ internal vertices if and only if $G'=G-x$ has a DFS tree with exactly $t$ internal vertices inductively for unlabeled vertices of $I$ of degree at least two, we obtain that Rule~\ref{dual_min_R2} is safe.

\medskip
Denote now by $G'$ the graph obtained from $G$ by the application of Rules~\ref{dual_min_R1} and \ref{dual_min_R2}. Because both rules are safe, for any integer $t\geq 0$, $G$ has a DFS tree with exactly $t$ internal vertices if and only if $G'$ has a DFS tree with exactly $t$ internal vertices. Because of Rule~\ref{dual_min_R1}, $G'-S$ has at most $2s$ pendant vertices. Rule~\ref{dual_min_R2} guarantees that $G'-S$ has at most $2s\binom{s}{2}=s^2(s-1)$ vertices of degree at least two. Then the total number of vertices of $G'$ is at most $s^2(s-1)+2s+s=s^2(s-1)+3s$.

\medskip
It is straightforward to see that  Rule~\ref{dual_min_R1} can be applied in $\Oh(sn)$ time and Rule~\ref{dual_min_R2} can be applied in $\Oh(s^2n)$ time. Therefore, the algorithm is polynomial. This concludes the proof.
\qed
\end{proof}
As a direct consequence of Lemma~\ref{lem:vc} we obtain that \textsc{Min-LLT} and \textsc{Max-LLT} admit polynomial kernels when parameterized by the vertex cover number of the input graph.

% \begin{theorem}\label{thm:vc}
% \textsc{Min LLT} and \textsc{Max-LLT} admit kernels with $\Oh(\tau^3)$ vertices when parameterized by the vertex cover number $\tau=\tau(G)$. 
% \end{theorem}

We are ready to prove our kernels parameterized by vertex cover.

\begin{proof}[Proof of Theorem~\ref{thm:main:vc}]
We show the theorem for \textsc{Min-LLT}; the arguments for \textsc{Max-LLT} are almost identical. 
Recall that the task of \textsc{Min-LLT} is  to decide,   given a graph $G$ and an integer $k\geq 0$, whether $G$ has a DFS tree with at most $k$ leaves. Equivalently, we can ask whether $G$ has a DFS tree with at least $|V(G)|-k$ internal vertices. Let $(G,k)$ be an instance of \textsc{Min-LLT}. We assume that $G$ is connected as, otherwise, $(G,k)$ is a no-instance and we can return a trivial no-instance of  \textsc{Min-LLT} of constant size. 

First, we find a vertex cover $S$ of $G$. For this, we apply a folklore approximation algorithm (see, e.g.,~\cite{cyganPC}) that greedily finds an inclusion-maximal matching $M$ in $G$ and takes the set $S$ of endpoints of the edges of $M$. It is well-known that $|S|\leq 2\tau$. Then we apply the algorithm from Lemma~\ref{lem:vc}. Let $G'$ be the output graph. By Lemma~\ref{lem:vc}, $G'$ has $\Oh(\tau^3)$ vertices. We set $k'=k-|V(G)|+|V(G')|$ and return the instance $(G',k')$ of \textsc{Min-LLT}.

Suppose that $G$ has a DFS tree with at most $k$ leaves. Then $G$ has a DFS tree with $t\geq |V(G)|-k$ internal vertices. By Lemma~\ref{lem:vc}, $G'$ also has a DFS tree with $t$ internal vertices. Then $G'$ has a DFS tree with $|V(G')|-t\leq |V(G')|-(|V(G)|-k)=k'$ leaves. For the opposite direction, assume that $G'$ has a DFS tree with at most $k'$ leaves. Then $G'$ has a DFS tree with $t\geq |V(G')|-k'=|V(G)|-k$ internal vertices. By Lemma~\ref{lem:vc}, $G$ has a DFS tree with $t$ internal vertices and, therefore, $G$ has a DFS tree with at most $k$ leaves.

Because $S$ can be constructed in linear time and the algorithm from  Lemma~\ref{lem:vc} is polynomial, the overall running time is polynomial. This concludes the proof.
\qed
\end{proof}
\noindent
Now we demonstrate a polynomial kernel for \textsc{Dual Min-LLT}.

\begin{theorem}\label{thm:dual-min}
 \textsc{Dual Min-LLT} admits a kernel with $\Oh(k^3)$ vertices.
\end{theorem}

\begin{proof}
Recall that the task of \textsc{Dual Min-DLL} is to verify, given a graph $G$ and an integer $k\geq 0$, whether $G$ has a DFS tree with  at most $n-k$ leaves. Equivalently, the task is to check whether  $G$ has a DFS tree with at least $k$ internal vertices. Let $(G,k)$ be an instance of \textsc{Dual Min-LLT}. If $G$ is disconnected, then $(G,k)$ is a no-instance and we return a trivial no-instance of \textsc{Dual Min-DLL}  of constant size. From now, we assume that $G$ is connected.

We select an arbitrary vertex $r$ of $G$ and run the DFS algorithm from this vertex. The algorithm produces a DFS tree $T$. Let $S$ be the set of internal vertices of $T$.  If $|S|\geq k$, then we conclude that $(G,k)$ is a yes-instance. Then the kernelization algorithm returns a trivial yes-instance of \textsc{Dual Min-LLT} of constant size and stops. Assume that this is not the case and $|S|\leq k-1$.

By Observation~\ref{obs:int}, we have that $S$ is a vertex cover of $G$ of size $s\leq k-1$. We use $S$ to call the algorithm from Lemma~\ref{lem:vc}. Let $G'$ be a graph produced by the algorithm. By Lemma~\ref{lem:vc}, $G'$ has $\Oh(k^3)$ vertices. Our kernelization algorithm returns $(G',k)$ and stops.

To see correctness, it is sufficient to observe that by Lemma~\ref{lem:vc}, for any integer $t\geq k$, $G$ has a DFS tree with $t$ internal vertices if and only if $G'$ has a DFS tree with $t$ internal vertices.
Because the DFS algorithm runs in linear time (see, e.g.,~\cite{cormen2009}) and the algorithm from Lemma~\ref{lem:vc} is polynomial, the overall running time is polynomial. This completes the proof.
\qed
\end{proof}

\noindent
We use similar arguments to prove the following theorem in Appendix \ref{appendix:B}.

\begin{theorem}\label{thm:dual-max}
 \textsc{Dual Max-LLT} admits a kernel with $\Oh(k^3)$ vertices.
\end{theorem}

Theorems~\ref{thm:dual-min} and~\ref{thm:dual-max} implies Theorem~\ref{thm:main:dual}.

\section{FPT Algorithms}\label{sectn:fpt_algorithms}
In this section, we give algorithms that solve \textsc{Dual Min-LLT} and \textsc{Dual Max-LLT} in FPT time using the kernels given in the previous section. Our algorithms are brute force algorithms which guess internal vertices. 

Recall that the standard DFS algorithm~\cite{cormen2009} outputs a labeled spanning tree. More formally, given an $n$-vertex graph and a root vertex $r$, the algorithm outputs a DFS tree $T$ rooted in $r$ and assigns to the vertices of $G$ distinct labels $d[v]$ from $\{1,\ldots,n\}$ giving the order in which the vertices were discovered by the algorithm. Thus, the algorithm outputs a linear ordering of vertices. Given an ordering $v_1,\ldots,v_n$ of $V(G)$, we say that a DFS tree $T$ \emph{respects} the ordering if $T$ is produced by the DFS algorithm in such a way that $d[v_i]=i$ for every $i\in\{1,\ldots,n\}$. Observe that for an ordering of the vertices of $G$, there is a unique way to run the DFS algorithm to obtain $T$ respecting the ordering. This gives us the following observation.

\begin{observation}\label{obs:order}
It can be decided in linear time, given an ordering $v_1,\ldots,v_n$ of the vertices of a graph $G$, whether $G$ has a DFS tree respecting the ordering. Furthermore, if such a tree $T$ exists, it is unique and can be constructed in linear time. 
\end{observation}
Let $G$ be a graph and let $r\in V(G)$. For a tree $T\subseteq G$ with $r\in V(T)$, we say that $T$ is \emph{extendable} to a DFS tree rooted in $r$, if there is a DFS tree $T'$ of $G$ rooted in $r$ such that $T$ is a subtree of $T'$. We call $T'$ an \emph{extension} of $T$. The definition of a DFS tree immediately gives us the following necessary and sufficient conditions for the extendability of $T$.    

\begin{observation}\label{obs:ext}
Let $G$ be a graph with $r\in V(G)$ and let $T\subseteq G$ be a tree containing $r$. Then $T$ is extendable to a DFS tree rooted in $r$ if and only if 
\begin{itemize}
\item[(i)] $T$ is a DFS tree rooted in $r$ of $G[V(T)]$,  
\item[(ii)] for every connected component $C$ of $G-V(T)$, the vertices of $N_G(V(C))$ are vertices of a root-to-leaf path of $T$.
\end{itemize}
\end{observation}
Note that (i) and (ii) can be verified in polynomial (in fact, linear) time.
We need the following variants of Observation~\ref{obs:ext} for special extensions in our algorithms.

\begin{observation}\label{obs:ext-int}
Let $G$ be a graph with $r\in V(G)$ and let $T\subseteq G$ be a tree containing $r$. Then $T$ is extendable to a DFS tree rooted in $r$ with an extension $T'$ such that the vertices of $V(T)$ are internal vertices of $T'$ if and only if 
\begin{itemize}
\item[(i)] $T$ is a DFS tree rooted in $r$ of $G[V(T)]$,  
\item[(ii)] for every connected component $C$ of $G-V(T)$, the vertices of $N_G(V(C))$ are vertices of a root-to-leaf path of $T$,
\item[(iii)] for every leaf $v$ of $T$, there is $u\in V(G)\setminus V(T)$ that is adjacent to $v$.
\end{itemize}
\end{observation}

\begin{observation}\label{obs:ext-leaves}
Let $G$ be a graph with $r\in V(G)$ and let $T\subseteq G$ be a tree containing $r$. Then $T$ is extendable to a DFS tree rooted in $r$ with an extension $T'$ such that the vertices of $L=V(G)\setminus V(T)$ are leaves of $T'$ if and only if 
\begin{itemize}
\item[(i)] $T$ is a DFS tree rooted in $r$ of $G[V(T)]$,  
\item[(ii)] $L$ is an independent set,
\item[(iii)] for every $v\in L$,  the vertices of $N_G(v)$ are vertices of a root-to-leaf path of $T$.
\end{itemize}
\end{observation}

%Now we are ready to describe our algorithms.
\noindent
Now, we are ready to describe our algorithms. For the proof of Lemma \ref{lem:XP}, see Appendix \ref{appendix:C}.

\begin{lemma}\label{lem:XP}
 \textsc{Dual Min-LLT} and \textsc{Dual Max-LLT} can be solved in $n^{\Oh(k)}$ time.  
 \end{lemma}

\noindent
Combining Lemma~\ref{lem:XP} and Theorem~\ref{thm:main:dual} implies Theorem~\ref{thm:main:FPT} by providing $k^{\Oh(k)}\cdot n^{\Oh(1)}$ time algorithms for the dual problems.

% \begin{theorem}\label{thm:FPT}
%  \textsc{Dual Min-LLT} and \textsc{Dual Max-LLT} can be solved in $k^{\Oh(k)}\cdot n^{\Oh(1)}$ time.  
%  \end{theorem}

\section{Conclusion}
\label{concln}
We have shown that \textsc{Dual Min-LLT} and \textsc{Dual Max-LLT} admit kernels with $\Oh(k^3)$ vertices and 
can be solved in $k^{\Oh(k)}\cdot n^{\Oh(1)}$ time.   
A natural question is whether the problems have linear kernels, such as for $k$-\textsc{Internal Spanning Tree}~\cite{wenjunLi2017}.
Another question is whether the problems can be solved by  single-exponential FPT algorithms.

As a byproduct of our kernelization algorithms for  \textsc{Dual Min-LLT} and \textsc{Dual Max-LLT}, we also proved that  \textsc{Min-LLT} and \textsc{Max-LLT} admit polynomial kernels for the structural parameterization by the vertex cover number. It is natural to wonder whether polynomial kernels exist for other structural parameterizations. In particular, it could be interesting to consider the parameterization by the \emph{feedback vertex} number, i.e., by the minimum size of a vertex set $X$ such that $G-X$ is a forest. 

\subsubsection{Acknowledgements} We acknowledge support from the Research Council of Norway grant ``Parameterized Complexity for Practical Computing (PCPC)'' (NFR, no. 274526) and ``Beyond Worst-Case Analysis in Algorithms (BWCA)'' (NFR, no. 314528).

\bibliographystyle{splncs04}
\bibliography{references}

\appendix
%%%%%%%%%%%%%%%%%%%%%%%%%%%%%%%%%%%%%%%%%%%%%%%%%%%%%
\newcommand{\qedclaim}{\hfill $\blacklozenge$}
\section {Proof of Claim~\ref{cl:aux} in the Proof of Lemma \ref{lem:vc}} \label{appendix:A}
\begin{proof}
We show (i) by contradiction. Assume that there are $u,v\in N_G(x)$ such that the lowest common ancestor $w$ of these vertices is distinct from $u$ and $v$. Because $x$ is not labeled by Rule~\ref{dual_min_R2}, $|W_{uv}|>2s$. Hence, by Lemma~\ref{upperB_internal}, there is a vertex $z\in W_{uv}$ such that $z$ is a leaf of $T$. However, any leaf in a DFS tree of $T$ can be adjacent only to its ancestors in $T$. This contradiction proves the claim.

We use exactly the same arguments to prove (ii) by replacing $T$ by $T'$ and observing that $S$ is a vertex cover of $G'$.

To show (iii), let $T'$ be a DFS tree with a root $r$. By (ii), there is a leaf $y$ such that the vertices of $N_G(x)$ are vertices of the $(r,y)$-path in $T'$. Observe that $y$ may be not unique. We prove that $y\notin N_G(x)$. For the sake of contradiction, assume that $x$ and $y$ are adjacent. Because $d_G(x)\geq 2$, $x$ has a neighbor $u\neq x$. Because $x$ is not labeled by Rule~\ref{dual_min_R2}, $|W_{uy}|>2s$. By Lemma~\ref{upperB_internal}, we obtain that there is $v\in W_{uy}$ that is a leaf of $T'$. We have that $vy\in E(G')$ but two leaves of a DFS tree cannot be adjacent; a contradiction. This proves that $y\notin N_G(x)$ and concludes the proof of the claim.  
\qedclaim
\end{proof}
%%%%%%%%%%%%%%%%%%%%%%%%%%%%%%%%%%%%%%%%%%%%%
\section{Proof of Theorem \ref{thm:dual-max}} \label{appendix:B}
\begin{proof}
The aim of \textsc{Dual Max-LLT} is to decide, given a graph $G$ and an integer $k\geq 0$, whether $G$ has a DFS tree with  at least $n-k$ leaves. This is equivalent to asking whether $G$ has a DFS tree with at most $k$ internal vertices. Let $(G,k)$ be an instance of \textsc{Dual Max-LLT}. If $G$ is disconnected, then $(G,k)$ is a no-instance, and we return a trivial no-instance of \textsc{Dual Max-DLL}  of constant size. From now, we assume that $G$ is connected.

If $T$ is a DFS tree, then the set of internal vertices of $T$ is a vertex cover of $G$ by Observation~\ref{obs:int}. Hence, if $G$ has a DFS tree with at most $k$ internal vertices, then $\tau(G)\leq k$. We approximate $\tau(G)$ by selecting greedily an inclusion-maximal matching $M$ in $G$ (see, e.g.,~\cite{cyganPC}). If $|M|>k$, then we conclude that $\tau(G)>k$ and return a trivial no-instance of \textsc{Dual Max-DLL}  of constant size. 
Assume that this is not the case. Then we take $S$ as the set of endpoints of the edges of $M$ and observe that $S$ is a vertex cover of size at most $2k$. We call the algorithm from Lemma~\ref{lem:vc} for $G$ and $S$, which outputs a graph $G'$ with $\Oh(k^3)$ vertices. The kernelization algorithm returns the instance $(G',k)$ of \textsc{Dual Max-DLL} and stops.

To see the correctness, note that by Lemma~\ref{lem:vc}, for any nonnegative integer $t\leq k$, $G$ has a DFS tree with $t$ internal vertices if and only if $G'$ has a DFS tree with $t$ internal vertices.
Because $M$ can be constructed in linear time and the algorithm from Lemma~\ref{lem:vc} is polynomial, the overall running time is polynomial. This completes the proof.
\qed
\end{proof}

\section{Proof of Lemma \ref{lem:XP} in Section \ref{sectn:fpt_algorithms}}
\label{appendix:C}
\begin{proof}
First, we give an algorithm for \textsc{Dual Min-LLT}. 
Let $(G,k)$ be an instance of the problem. If $G$ is disconnected, then $(G,k)$ is a no-instance. Assume that this is not the case. Also, we have a trivial no-instance if $n\leq k$ and we assume that $n\geq k$.

Recall that the equivalent task of \textsc{Dual Min-LLT} is to decide, given a graph $G$ and an integer $k$, whether $G$ has a DFS tree with at least $k$ internal vertices. 
We guess a set $S$ of $k$ internal vertices containing a root of a solution DFS tree $T$ forming a subtree $T'=T[S]$. To guess $T'$ and $S$, we apply Observation~\ref{obs:order} using the fact that $T'$ should be a DFS tree of $G[S]$.
Formally, we consider all $k$-tuples $(v_1,\ldots,v_k)$ of distinct vertices of $G$. For each $k$-tuple, we check whether there is a DFS tree $T'$ of $G[S]$, where $S=\{v_1,\ldots,v_k\}$, respecting the ordering $v_1,\ldots,v_k$ using Observation~\ref{obs:order}.  If such a tree $T'$ exists, we use Observation~\ref{obs:ext-int} to check  whether $T'$ has an extension $T$ such that the vertices of $S$ are internal vertices of $T$. If we find such a $k$-tuple, we conclude that 
$(G,k)$ is a yes-instance of   \textsc{Dual Min-LLT}. Otherwise, if we fail to find $T'$ and a required extension for all $k$-tuples, we conclude that   $(G,k)$ is a no-instance of   \textsc{Dual Min-LLT}.
The correctness of the algorithm immediately follows from Observations~\ref{obs:order} and \ref{obs:ext-int}. Because we have at most $n^k$ $k$-tuples of vertices, we obtain that the overall running time is $n^{\Oh(k)}$.

We use a similar strategy for \textsc{Dual Max-LLT}. Recall that now the task is to decide whether a graph $G$ has a DFS tree with at most $k$ internal vertices. 
Let $(G,k)$ be an instance of the problem. As above, we can assume that $G$ is connected. Also, if $n\leq k$, then $(G,k)$ is a yes-instance and we can assume that $n>k$.
We guess a set $S$ of $k$  vertices containing a root and the internal vertices of a solution DFS tree $T$ and a subtree $T'=T[S]$.  For this, we consider all 
$k$-tuples $(v_1,\ldots,v_k)$ of distinct vertices of $G$. For each $k$-tuple, we check whether there is a DFS tree $T'$ of $G[S]$, where $S=\{v_1,\ldots,v_k\}$, respecting the ordering $v_1,\ldots,v_k$ using Observation~\ref{obs:order}.  If such a tree $T'$ exists, we use Observation~\ref{obs:ext-leaves} to check whether $T'$ has an extension $T$ such that the vertices of $V(G)\setminus S$ are leaves of $T$. If we find such a $k$-tuple, we conclude that 
$(G,k)$ is a yes-instance of   \textsc{Dual Max-LLT}. Otherwise, if we fail to find $T'$ and a required extension for all $k$-tuples, we conclude that   $(G,k)$ is a no-instance of   \textsc{Dual Max-LLT}.
Observations~\ref{obs:order} and \ref{obs:ext-leaves} imply correctness, and  the overall running time is $n^{\Oh(k)}$.
This concludes the proof.
\qed
\end{proof}
\end{document}